\documentclass[envcountsame,envcountsect]{llncs}

\usepackage{tikz}
\usetikzlibrary{arrows,shapes,trees}
\usepackage{latexsym}
\usepackage{amssymb}  
\usepackage{amsmath}
\usepackage{stmaryrd}

\newenvironment{propositionAppendix}[1]{{\noindent\bf Proposition \ref{#1}.}\it}{}

\begin{document}

\title{Three Variables Suffice for Real-Time Logic}
\author{Timos Antonopoulos\inst{1} \and Paul Hunter\inst{2} \and Shahab Raza\inst{1} \and James Worrell\inst{1}}

\institute{Department of Computer Science, Oxford University, UK\\
\email{\{timos,shahab,jbw\}@cs.ox.ac.uk}
\and
D\'{e}partement d'Informatique, Universit\'{e} Libre de Bruxelles, Belgium
\email{paul.hunter@ulb.ac.be}}

\maketitle

\begin{abstract}
  A natural framework for real-time specification is monadic
  first-order logic over the structure $(\mathbb{R},<,+1)$---the
  ordered real line with unary $+1$ function.  Our main result is that
  $(\mathbb{R},<,+1)$ has the 3-variable property: every monadic
  first-order formula with at most 3 free variables is equivalent over
  this structure to one that uses 3 variables in total.  As a
  corollary we obtain also the 3-variable property for the structure
  $(\mathbb{R},<,f)$ for any fixed linear function
  $f:\mathbb{R}\rightarrow\mathbb{R}$.  On the other hand, we exhibit a
  countable dense linear order $(E,<)$ and a bijection $f:E\rightarrow
  E$ such that $(E,<,f)$ does not have the $k$-variable property for
  any $k$.
\end{abstract}
\section{Introduction}
Monadic first-order logic is an expansion of first-order logic by
infinitely many unary predicate variables.  In this setting a class of
structures $\mathcal{C}$ is said to have the \emph{$k$-variable
  property} if every formula with at most $k$ free first-order
variables is equivalent over $\mathcal{C}$ to a formula with at most
$k$ first-order variables in total (allowing multiple binding
occurrences of the same variable).  The $k$-variable property for
monadic first-order logic over linearly ordered structures has been
studied
in~\cite{Dawar05,Gabbay81,GroheS05,HodkinsonS97,ImmermanK89,Poizat82,Rossman08},
among others.  In finite model theory the $k$-variable property plays
an important role in descriptive complexity.  Over infinite models it
is closely connected with expressive completeness of temporal logics.

It is well known that Linear Temporal Logic (LTL) with Stavi
modalities is expressively complete for monadic first-order logic over
the class of linear orders~\cite{GabbayPSS80,Kamp68}.  More precisely,
LTL is expressively complete for the class of monadic first-order
formulas with one free variable (corresponding to the fact that LTL
formulas are evaluated at a single point of a linear order).  The
translation from LTL to first-order logic is a straightforward
inductive construction that maps into the 3-variable fragment of
first-order logic.  It follows that every monadic first-order formula
with at most one free variable is equivalent to a 3-variable formula
over linear orders.  However this is a strictly weaker condition than
the 3-variable property in general: Hodkinson and
Simon~\cite{HodkinsonS97} give a class of partial orders over which
every monadic first-order formula with at most one free variable is
equivalent to a 3-variable formula, but which does not have the
$k$-variable property for any $k$.  Nevertheless the 3-variable
property \emph{does} hold over linear orders, as shown by
Poizat~\cite{Poizat82} and Immerman and Kozen~\cite{ImmermanK89},
using Ehrenfeucht-Fra\"{i}ss\'{e} games.

Going beyond pure linear orders, Venema~\cite{Venema90} gives a dense
linear order with a single equivalence relation over which monadic
first-order logic does not have the $k$-variable property for any $k$.
A more powerful result by Rossman~\cite{Rossman08} shows that the
class of finite linearly ordered graphs does not have the $k$-variable
property for any $k$, resolving a longstanding conjecture of
Immerman~\cite{Immerman82}.

In this paper we are concerned with monadic first-order logic over the
ordered reals with unary $+1$ function $(\mathbb{R},<,+1)$.  This
logic has been extensively studied in the context of real-time
verification.  An expansion of $(\mathbb{R},<,+1)$ with
interpretations of the unary predicate variables can be seen as a
real-time signal, with the unary predicates denoting propositions that
may or may not hold at any given time.  First-order logic over signals
can express both metric and order-theoretic temporal properties and is
an expressive meta-language into which many different real-time logics
can directly be translated~\cite{HirshfeldR12,HirshfeldR05}.  In
particular, first-order logic over signals is expressively equivalent
with Metric Temporal Logic (MTL)~\cite{Hunter13,HunterOW13}.

Our main result is that $(\mathbb{R},<,+1)$ has the $3$-variable
property.  For example, the property 
\[ \forall x_1 \exists x_2 \exists x_3 \exists x_4 
\left(x_4<x_1+1 \wedge
\bigwedge_{1\leq i\leq 3} x_i < x_{i+1} \wedge \bigwedge_{2\leq i \leq 4}
P(x_i) \right) \]
that $P$ is true at least $3$ times in every unit interval
can equivalently be written
\begin{align*}
\forall x \exists y  \big(x<y \wedge P(y) \wedge 
\exists z(y<z\wedge P(z) \wedge \exists y(z<y<x+1 \wedge P(y)))) \, .
\end{align*}

From the expressive completeness of MTL it follows that every monadic
first-order formula with at most one free variable is equivalent to a
3-variable formula over $(\mathbb{R},<,+1)$.  However, as remarked
above, this condition is weaker than the 3-variable property in
general.  Moreover the proof of expressive completeness of MTL
combines intricate syntactic manipulations of MTL formulas together
with technically involved results of~\cite{GabbayPSS80} for LTL\@.  On
the other hand, the model-theoretic argument given here, using
Ehrenfeucht-Fra\"{i}ss\'{e} games, is self-contained and exposes a
novel two-level compositional technique that can potentially be
applied in more general settings and to other ends (see the
Conclusion).

As a corollary of our main result we straightforwardly derive the
3-variable property for each structure $(\mathbb{R},<,f)$ with
$f:\mathbb{R}\rightarrow\mathbb{R}$ a linear function $f(x)=ax+b$.  We
believe that the result can be generalised to other linear orders and
suitably well-behaved functions.  However, unsurprisingly, the
property fails for sufficiently `wild' functions.  Adapting Venema's
construction~\cite{Venema90}, we give an example of a countable dense
linear order $E$ and a (far from monotone) bijection $f:E\rightarrow
E$ such that $(E,<,f)$ does not have the $k$-variable property for any
$k$.

The paper naturally divides into two parts.  Sections~\ref{sec:local}
to \ref{sec:main} are exclusively concerned with the
structure $(\mathbb{R},<,+1)$, while Sections~\ref{sec:arithmetic} and
\ref{sec:counterexample} consider other unary functions in place of
$+1$.

\section{Background}
\subsection{Ehrenfeucht-Fra\"{i}ss\'{e} Games}
\label{sec:EF}
Throughout the paper we work with a first-order signature $\sigma$ 
with a binary relation symbol $<$ and a unary function symbol $f$.
The monadic first-order language over $\sigma$ is defined as follows:
\begin{itemize}
\item There is an infinite collection of monadic predicate variables
  $P_1,P_2,\ldots$.
\item The atomic formulas are $x=y$, $x<y$, $P_n(x)$, and $x=f(y)$ for
first-order variables $x$ and $y$ and $n\in \mathbb{N}$.  
\item If $\varphi_1$ and $\varphi_2$ are formulas and $x$ is a
  variable then $\neg \varphi_1$, $\varphi_1\wedge \varphi_2$ and
  $\exists x\, \varphi_1$ are also formulas.
\end{itemize}

Referring to the restricted use of the function symbol $f$ in atomic
formulas, we say that the formulas above are \emph{unnested}.
The unnesting assumption essentially amounts to treating the function
symbol $f$ as a binary relation symbol.  We make this assumption as an
alternative to restricting to a purely relational signature.  The
unnesting assumption does not affect expressiveness since we can
translate an arbitrary formula to an equivalent unnested formula by
successively replacing atomic formulas $f^m(x)=f^n(y)$ with $m>0$ by
$\exists z\, (z=f(x) \wedge f^{m-1}(z)=f^n(y))$, and similarly for
$f^m(x)<f^n(y)$.  While this transformation may increase the
quantifier depth, it preserves the subclass of 3-variable formulas.

Let $\mathbf{A}=(A,<^\mathbf{A},f^\mathbf{A},\overline{P}^\mathbf{A})$ denote a
$\sigma$-structure expanded with interpretations of the monadic
predicate variables $P_1,P_2,\ldots$.  We call $\mathbf{A}$ a
\emph{labelled $\sigma$-structure}.  Given first-order variables
$x_1,\ldots,x_k$, an \emph{assignment} in $\mathbf{A}$ with domain
$\{x_1,\ldots,x_k\}$ is a tuple $\overline{u}=u_1 \ldots u_k$ in
$A^k$. Given another assignment $\overline{v}$ with the same domain in
a labelled $\sigma$-structure $\mathbf{B}$, we say that
$(\overline{u},\overline{v})$ is a \emph{partial isomorphism} between
$\mathbf{A}$ and $\mathbf{B}$ if $\mathbf{A} \models \varphi[\overline{u}]$ iff
$\mathbf{B} \models \varphi[\overline{v}]$ for all atomic formulas
$\varphi(x_1,\ldots,x_k)$.

The \emph{Ehrenfeucht-Fra\"{i}ss\'{e}} (EF) game on structures
$\mathbf{A}$ and $\mathbf{B}$ is played by two
players---\emph{Spoiler} and \emph{Duplicator}.\footnote{By
  convention, Spoiler is male and Duplicator is female.}  Each player
has a collection of pebbles, respectively labelled $x_1,x_2,\ldots$.
The game is played over a fixed number of rounds.  In each round
Spoiler chooses a structure and places a pebble on an element of the
structure (either an unused pebble or one that has already been
placed); Duplicator responds by placing a pebble with the same label
on some element of the other structure.  A placement of $k$ pebbles on
each structure naturally determines a pair of assignments
$(\overline{u},\overline{v})$, called a \emph{$k$-configuration}.
(Our notation for $k$-configurations leaves the structures
$\mathbf{A}$ and $\mathbf{B}$ implicit.)  If the configuration after
each round is a partial isomorphism then Duplicator wins, otherwise
Spoiler wins.  For each configuration $(\overline{u},\overline{v})$
and number of rounds $n$, exactly one of the players has a winning
strategy in the $n$-round game starting from
$(\overline{u},\overline{v})$ (see~\cite{ImmermanK89} for more
details).

A natural restriction on Ehrenfeucht-Fra\"{i}ss\'{e} games is to limit
each player to a fixed number of pebbles.  In the \emph{$k$-pebble
  game} both Spoiler and Duplicator possess only $k$ pebbles,
respectively labelled $x_1,\ldots,x_k$.  The following theorem shows
how Ehrenfeucht-Fra\"{i}ss\'{e} games can be used to characterise the
expressiveness of first-order logic according to the number of
variables.

\begin{theorem}[\cite{ImmermanK89}]
Let $\mathcal{C}$ be a class of $\sigma$-structures such that for all
$n$ there exists $m$ such that if Spoiler wins the $n$-round
Ehrenfeucht-Fra\"{i}ss\'{e} game on a pair of labelled structures from
$\mathcal{C}$ starting in a $k$-configuration
$(\overline{u},\overline{v})$, then he also wins the $m$-round
$k$-pebble game starting in $(\overline{u},\overline{v})$.  Then
$\mathcal{C}$ has the $k$-variable property.
\label{thm:EF}
\end{theorem}


In the remainder of this section we specialise our attention to the
$\sigma$-structure $(\mathbb{R},<,+1)$.  In this case we call a
labelled $\sigma$-structure a \emph{signal}.

In addition to $k$-pebble games, on signals we introduce another
restriction of Ehrenfeucht-Fra\"{i}ss\'{e} games.  Given an assignment
$\overline{u} \in \mathbb{R}^k$ with domain $\{x_1,\ldots,x_k\}$, the
\emph{diameter} of $\overline{u}$ is $\mathit{diam}(\overline{u}) =
\max\{ |u_i - u_j|: 1\leq i,j \leq k\}$.  Given $D \in \mathbb{R}$,
the \emph{$D$-local game} on a pair of signals is such that Spoiler
and Duplicator must maintain the invariant that all assignments have
diameter at most $D$.

We will always explicitly indicate any restrictions on the number of
pebbles or the diameter of configurations in games: thus the default
notion of Ehrenfeucht-Fra\"{i}ss\'{e} game is without restriction on
the number of pebbles or the diameter.

Recall that our main result is that $(\mathbb{R},<,+1)$ has the
3-variable property.  The main conceptual insight underlying the proof
is that one should first prove the 3-variable property for ``local''
formulas.  We treat locality semantically through the notion of local
EF games, as defined above, but intuitively a local formula is one
that asserts properties of elements at a bounded distance from one
another.  For example, $\exists x\exists y\,(P(x) \wedge Q(y) \wedge
x,<y<x+1)$ is local but $\exists x\exists y\, (P(x) \wedge Q(y))$ is
not local.

We prove the 3-variable property for local formulas by a compositional
argument based on the fractional-part preorder on $\mathbb{R}$.  We
then extend the 3-variable property to all formulas by adapting the
well-known composition lemma for sums of linear orders to the
structure $(\mathbb{R},<,+1)$.  Roughly speaking, this second
compositional lemma shows that Duplicator strategies on summands can
be composed provided that there is sufficient distance between pebbles
in different summands.  However this precondition is not always met
and here it is crucial that we have already established the 3-variable
property for local formulas.

\subsection{Interpretations}
\label{sec:interpretation}
In this section we briefly deviate from the setting of linear orders
and unary functions to recall from~\cite[Chapter 4.3]{Hodges} the
notion of an interpretation of one first-order structure in another.

Let $\sigma_1$ and $\sigma_2$ be signatures, $\mathbf{A}$ a
$\sigma_1$-structure with domain $A$, $\mathbf{B}$ a $\sigma_2$-structure with domain
$B$, and $n$ a positive integer.  An $n$-dimensional
\emph{interpretation} $\Gamma$ of $\mathbf{B}$ in $\mathbf{A}$
consists of three items:
\begin{itemize}
\item a $\sigma_1$-formula $\partial_\Gamma(x_1,\ldots,x_n)$ denoting
the \emph{domain} of the interpretation, which is the set $\partial_\Gamma(A^n) := \{
  \overline{a}\in A^n : \mathbf{A} \models
  \partial_\Gamma[\overline{a}]\}$.
\item for each unnested atomic $\sigma_2$-formula
  $\varphi(x_1,\ldots,x_m)$, a $\sigma_1$-formula
  $\varphi_\Gamma(\overline{x}_1,\ldots,\overline{x}_m)$ in which the
  $\overline{x}_i$ are disjoint $n$-tuples of distinct variables,
\item a surjective \emph{coding map} $f_\Gamma : \partial_\Gamma(A^n) \rightarrow
  B$ such that for all unnested atomic $\sigma_2$-formulas $\varphi$ and
  all $\overline{a}_i \in \partial_\Gamma(A^n)$,
\[ \mathbf{B} \models \varphi[f_\Gamma \overline{a}_1,\ldots,f_\Gamma \overline{a}_m] \mbox{ iff }
    \mathbf{A} \models \varphi_\Gamma
    [\overline{a}_1,\ldots,\overline{a}_m] \, .\]
\end{itemize}
\section{From Local Games to 3-Pebble Games}
\label{sec:local}
In this section we consider an Ehrenfeucht-Fra\"{i}ss\'{e} game on two
signals $\mathbf{A}$ and $\mathbf{B}$.  Here $\overline{u}=u_1\ldots
u_s$ will always denote an assignment in $\mathbf{A}$ and
$\overline{v}=v_1\ldots v_s$ will always denote an assignment in
$\mathbf{B}$.

Write $\overline{u} \equiv \overline{v}$ if $\overline{u}$ and
$\overline{v}$ are indistinguishable by difference constraints, that
is, $u_i-u_j < c \Leftrightarrow v_i-v_j < c$ and $u_i-u_j = c
\Leftrightarrow v_i-v_j = c$ for all constants $c\in \mathbb{Z}$ and
indices $1\leq i,j \leq s$.  Equivalently, $\overline{u} \equiv
\overline{v}$ if and only if $\lfloor u_i-u_j \rfloor = \lfloor
v_i-v_j \rfloor$ for all indices $1 \leq i,j \leq s$.\footnote{Note
  that $\lceil u_i-u_j \rceil = \lceil v_i-v_j \rceil$ if and only if
  $\lfloor u_j-u_i \rfloor = \lfloor v_j-v_i \rfloor$, so there is no
  need to add a separate clause for ceiling in the characterisation of
  $\equiv$.}  Assignments that are indistinguishable by difference
constraints are, in particular, ordered the same way.

 Define the \emph{fractional part} of $u\in \mathbb{R}$ by
 $\mathrm{frac}(u)=u-\lfloor u \rfloor$. The proof of the following
 proposition can be found in the Appendix.
\begin{proposition}\label{prop:relative-frac}
Let $\overline{u}=u_1\ldots u_s$ and $\overline{v}=v_1\ldots v_s$ be
two assignments with $\overline{u}\equiv\overline{v}$. Then 
\[\mathrm{frac}(u_i-u_k)<\mathrm{frac}(u_j-u_k)
   \;\Leftrightarrow\;
   \mathrm{frac}(v_i-v_k)<\mathrm{frac}(v_j-v_k)\]
for all indices $i,j,k \in \{1,\ldots,s\}$.
\end{proposition}

 We say that $u_1\ldots u_s$ is in \emph{increasing order} if
 $\mathrm{frac}(u_i-u_1) \leq \mathrm{frac}(u_{i+1}-u_1)$ for
 $i=1,\ldots,s-1$.  Intuitively $u_1\ldots u_s$ is in increasing order
 if it is listed in increasing order of fractional parts relative to
 $u_1$.  Note that if $u_1\ldots u_s$ is in increasing order then any
 cyclic permutation is also in increasing order.  By
 Proposition~\ref{prop:relative-frac}, if $\overline{u}\equiv
 \overline{v}$ then $\overline{u}$ and $\overline{v}$ can both be
 brought into increasing order by a common permutation.


The following proposition can be seen as a compositional lemma for
$\equiv$.  The proof can be found in the Appendix.

\begin{proposition}
	Suppose that $u_1 \ldots u_s$ and $v_1 \ldots v_s$ are both increasing
	and that $u_1 \ldots u_m \equiv v_1 \ldots v_m$ and
	$u_m \ldots u_s \equiv v_m \ldots v_s$ for some $m$, $1 \leq
	m \leq s$.  Then $u_1 \ldots u_s \equiv v_1 \ldots v_s$.
\label{prop:composition}
\end{proposition}


Proposition~\ref{prop:metric} and Corollary~\ref{corl:metric} show
that three pebbles suffice to determine equivalence of configurations
under the relation $\equiv$.

\begin{proposition}
Let $n \in \mathbb{N}$.  Consider a 2-configuration $(u_1u_2,v_1v_2)$
such that either (i)~$u_1-u_2 < c$ and $v_1-v_2 \not< c$ for some
non-negative integer $c<2^n$ or (ii)~$u_1-u_2=c$ and $v_1-v_2 \neq c$
for some non-negative integer $c\leq 2^n$.  Then Spoiler wins the
$n$-round 3-pebble game from $(u_1u_2,v_1v_2)$.
\label{prop:metric}
\end{proposition}

\begin{proof}
The proof is by induction on $n$.

Base case ($n=0$). Under either assumption (i) or (ii) the
configuration $(u_1u_2,v_1v_2)$ is not a partial isomorphism, and is
therefore immediately winning for Spoiler in the 3-pebble game.

Induction step ($n\geq 1$). Suppose $u_1-u_2<c$ but $v_1-v_2 \not< c$,
where $c<2^n$.  Write $c'=\lfloor c/2 \rfloor$, so that $c'<2^{n-1}$
and $c-c' \leq 2^{n-1}$.  Suppose that Spoiler places a pebble on
$u_3$ such that $u_1-u_3=c-c'$ and $u_3-u_2 < c'$.  Since $v_1-v_2
\not< c$, for any response $v_3$ of Duplicator we either have
$v_1-v_3\neq c-c'$ or $v_3-v_2 \not< c'$.  In the first case, by the
induction hypothesis, $(u_1u_3,v_1v_3)$ is winning in $n-1$ rounds for
Spoiler; likewise in the second case $(u_2u_3,v_2v_3)$ is winning in
$n-1$ rounds for Spoiler.  Thus in either case $(u_1u_2u_3,v_1v_2v_3)$
is winning in $n-1$ rounds for Spoiler.  We conclude that
$(u_1u_2,v_1v_2)$ is winning in $n$ rounds for Spoiler.  This handles
(i); Case (ii) is almost identical.  \qed
\end{proof}

\begin{corollary}
Let $(\overline{u},\overline{v})$ be a $3$-configuration such that
$\overline{u} \not\equiv \overline{v}$ and at least one of
$\overline{u}$ and $\overline{v}$ has diameter at most $2^m$.  Then
Spoiler wins the $m$-round 3-pebble game from
$(\overline{u},\overline{v})$.
\label{corl:metric}
\end{corollary}
\begin{proof}
  Since $\overline{u} \not\equiv \overline{v}$, there are indices $i,j$
  such that $u_i-u_j \sim c$ and $v_i-v_j \not\sim c$ for some
  non-negative integer constant $c$ and comparison operator
  $\mathop{\sim}\in\{<,=\}$.  Moreover, since at least one of
  $\overline{u}$ and $\overline{v}$ has diameter at most $2^m$, we can
  assume that $c \leq 2^m$.  But then Spoiler wins the $m$-round
  3-pebble game from $(\overline{u},\overline{v})$ by
  Proposition~\ref{prop:metric}. \qed
\end{proof}

One can think of following proposition as showing the 3-variable
property for local formulas.  The proof uses the compositional
principle in Proposition~\ref{prop:composition}.

\begin{proposition}
Let $(\overline{u},\overline{v})$ be a 3-configuration of diameter at
most $2^m$.  If Spoiler
wins the $n$-round $2^m$-local game from $(\overline{u},\overline{v})$
then he wins the $(m+n)$-round 3-pebble game from
$(\overline{u},\overline{v})$.
\label{prop:local}
\end{proposition}
\begin{proof}
If $\overline{u}\not\equiv\overline{v}$ then the result follows from
Corollary~\ref{corl:metric}.  Thus it suffices to prove the
proposition under the assumption $\overline{u}\equiv\overline{v}$.

Without loss of generality assume that $\overline{u}$ and
$\overline{v}$ are both increasing.  The proof is by induction on $n$,
with the following induction hypothesis.

\emph{Induction Hypothesis:} Let assignments $u_1\ldots u_s \equiv
v_1\ldots v_s$ be increasing and have diameter at most $2^m$.  If
Spoiler wins the $n$-round $2^m$-local game from
$(\overline{u},\overline{v})$, then he wins the $(m+n)$-round 3-pebble
game from a 2-configuration of the form $(u_iu_{i+1},v_iv_{i+1})$,
$1\leq i \leq s-1$, or $(u_su_1,v_sv_1)$.

\emph{Base case $(n=0)$.} By assumption $(\overline{u},\overline{v})$
is immediately winning for Spoiler in the local game.  Since
$\overline{u} \equiv \overline{v}$, $u_i$ and $v_i$ must disagree on a
unary predicate for some index $i$.  Then $(u_i,v_i)$ is immediately
winning for Spoiler in the 3-pebble game.  Clearly the position
remains immediately winning for Spoiler if we add an extra pebble to
each assignment.  Thus the base case of the induction is established.

\emph{Induction step $(n\geq 1)$.}  Pick a Spoiler move according to
his winning strategy in the local game in configuration
$(\overline{u},\overline{v})$.  Without loss of generality, assume
that this move, say $u'$, is in structure $\mathbf{A}$.  Since any
cyclic permutation of an increasing configuration is also increasing,
we may assume without loss of generality that $u_1\ldots u_su'$ is
increasing.

If $(u_1u_s,v_1v_s)$ is winning for Spoiler in the $(m+n)$-round
3-pebble game then we are done, so suppose that this is not the case.
Then there exists a Duplicator move $v'$ such that
$(u_1u_su',v_1v_sv')$ is winning for Duplicator in the $(m+n-1)$-round
3-pebble game.  Since $\mathrm{diam}(u_1u_su') \leq 2^m$, by
Corollary~\ref{corl:metric} we must have $u_1u_su' \equiv v_1v_sv'$.
It follows that $v_1\ldots v_sv'$ is increasing.

Since $u_1\ldots u_su'$ and $v_1\ldots v_sv'$ are increasing,
$u_1\ldots u_s\equiv v_1\ldots v_s$, and 
$u_su' \equiv v_sv'$, by Proposition~\ref{prop:composition} we have
\begin{gather}
u_1\ldots u_su' \equiv 
v_1\ldots v_sv' \, .
\label{eq:extend}
\end{gather}

Since the pair of assignments in (\ref{eq:extend}) is winning for
Spoiler in the $(n-1)$-round local game, by the induction hypothesis
there exists a sub-configuration (comprising two consecutive pebbles
in each assignment) from which Spoiler wins the $(m+n-1)$-round
3-pebble game.  This 2-configuration cannot be $(u_su',v_sv')$ nor
$(u_1u',v_1v')$, since $(u_1u_su',v_1v_sv')$ is winning for Duplicator
in the $(m+n-1)$-round 3-pebble game.  Thus Spoiler must win the
$(m+n-1)$-round 3-pebble game from a 2-configuration
$(u_iu_{i+1},v_iv_{i+1})$ for some $i\in \{1,\ldots,s-1\}$.  \emph{A
  fortiori} Spoiler also wins the $(m+n)$-round 3-pebble game from
this configuration.  \qed
\end{proof}

\section{Main Results}
\label{sec:main}
\subsection{Composition Lemma}
In this section we consider an Ehrenfeucht-Fra\"{i}ss\'{e} game on two
signals $\mathbf{A}$ and $\mathbf{B}$.  We will prove a Composition
Lemma that allows us to compose winning Duplicator strategies under
certain assumptions.  From this we obtain our main result, that
monadic first-order logic over signals has the 3-variable property.

Assume assignments $\overline{u}=u_1\ldots u_s$ in $\mathbf{A}$ and
$\overline{v} = v_1 \ldots v_s$ in $\mathbf{B}$ with $u_1<\ldots<u_s$
and $v_1 < \ldots < v_s$.  The Composition Lemma is predicated on a
decomposition of $\overline{u}$ into a \emph{left part}
$\overline{u}_{\triangleleft}=u_1\ldots u_l$, \emph{middle part}
$\overline{u}_{\diamond}=u_l\ldots u_r$, and \emph{right part}
$\overline{u}_{\triangleright}=u_r\ldots u_s$, where $1 \leq l \leq r
\leq s$.  We call $u_l$ the \emph{left boundary} and $u_r$ the
\emph{right boundary}.  The \emph{left margin} is defined to be
$\mathit{margin}(\overline{u}_{\triangleleft})=u_l-u_{l-1}$, where
$u_0=-\infty$ by convention.  Likewise the \emph{right margin} is
defined to be
$\mathit{margin}(\overline{u}_{\triangleright})=u_{r+1}-u_r$, where
$u_{s+1}=\infty$ by convention.  We consider a corresponding
decomposition of $\overline{v}$ into
$\overline{v}_{\triangleleft}=v_1\ldots v_l$,
$\overline{v}_{\diamond}=v_l\ldots v_r$, and
$\overline{v}_{\triangleright}=v_r\ldots v_s$, for the same values of
$l$ and $r$.

The Composition Lemma gives conditions under which we can obtain a
winning strategy for Duplicator in a configuration
$(\overline{u},\overline{v})$ by composing winning Duplicator
strategies in the \emph{left configuration}
$(\overline{u}_{\triangleleft},\overline{v}_{\triangleleft})$, the
\emph{middle configuration}
$(\overline{u}_{\diamond},\overline{v}_{\diamond})$, and \emph{right
  configuration}
$(\overline{u}_{\triangleright},\overline{v}_{\triangleright})$, see
Figure~\ref{fig:comp}.  The main idea behind the proof is to maintain
adequate separation between pebbles played by the left and middle
Duplicator strategies, and likewise between pebbles played by the
middle and right strategies.  We do this by maintaining the left and
right margins appropriately.  Importantly for later use, we need only
assume that Duplicator has a local winning strategy in the middle
configuration.
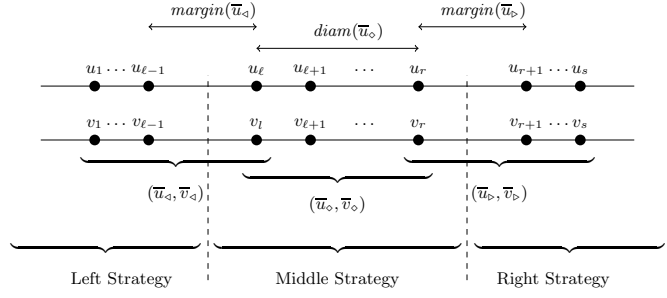
\begin{figure}[t]
			\centering
			\resizebox{9.0cm}{!}{
		  		\begin{tikzpicture}
					\draw[] (0,9) -- (11,9);
	
					\node [circle,fill=black,draw=none,inner sep=2pt] (x1) at (1,9) {};
					\node [draw=none,inner sep=1pt] () at (1,9.3) {$u_1$};
	
					\node [draw=none,inner sep=1pt] () at (1.4,9.3) {$\ldots$};
					
					\node [circle,fill=black,draw=none,inner sep=2pt] (x2) at (2,9) {};
					\node [draw=none,inner sep=1pt] () at (2,9.3) {$u_{\ell-1}$};
	
					\node [circle,fill=black,draw=none,inner sep=2pt] (x3) at (4,9) {};
					\node [draw=none,inner sep=1pt] () at (4,9.3) {$u_{\ell}$};
					
					\node [circle,fill=black,draw=none,inner sep=2pt] (x3) at (5,9) {};
					\node [draw=none,inner sep=1pt] () at (5,9.3) {$u_{\ell+1}$};
					
					\node [draw=none,inner sep=1pt] () at (6,9.3) {$\ldots$};
	
					\node [circle,fill=black,draw=none,inner sep=2pt] (x2) at (7,9) {};
					\node [draw=none,inner sep=1pt] () at (7,9.3) {$u_r$};
					
					\node [circle,fill=black,draw=none,inner sep=2pt] (x2) at (9,9) {};
					\node [draw=none,inner sep=1pt] () at (9,9.3) {$u_{r+1}$};
	
					\node [draw=none,inner sep=1pt] () at (9.6,9.3) {$\ldots$};
					
					\node [circle,fill=black,draw=none,inner sep=2pt] (x3) at (10,9) {};
					\node [draw=none,inner sep=1pt] () at (10,9.3) {$u_s$};

					\draw[] (0,8) -- (11,8);
	
					\node [circle,fill=black,draw=none,inner sep=2pt] (x1) at (1,8) {};
					\node [draw=none,inner sep=1pt] () at (1,8.3) {$v_1$};
	
					\node [draw=none,inner sep=1pt] () at (1.4,8.3) {$\ldots$};
					
					\node [circle,fill=black,draw=none,inner sep=2pt] (x2) at (2,8) {};
					\node [draw=none,inner sep=1pt] () at (2,8.3) {$v_{\ell-1}$};
	
					\node [circle,fill=black,draw=none,inner sep=2pt] (x3) at (4,8) {};
					\node [draw=none,inner sep=1pt] () at (4,8.3) {$v_l$};
					
					\node [circle,fill=black,draw=none,inner sep=2pt] (x3) at (5,8) {};
					\node [draw=none,inner sep=1pt] () at (5,8.3) {$v_{\ell+1}$};
					
					\node [draw=none,inner sep=1pt] () at (6,8.3) {$\ldots$};
	
					\node [circle,fill=black,draw=none,inner sep=2pt] (x2) at (7,8) {};
					\node [draw=none,inner sep=1pt] () at (7,8.3) {$v_r$};
					
					\node [circle,fill=black,draw=none,inner sep=2pt] (x2) at (9,8) {};
					\node [draw=none,inner sep=1pt] () at (9,8.3) {$v_{r+1}$};
	
					\node [draw=none,inner sep=1pt] () at (9.6,8.3) {$\ldots$};
					
					\node [circle,fill=black,draw=none,inner sep=2pt] (x3) at (10,8) {};
					\node [draw=none,inner sep=1pt] () at (10,8.3) {$v_s$};
					
					\node [draw=none] () at (8.5,7.6) {$\underbrace{\hspace{100pt}}$};
					\node [draw=none] () at (5.5,7.3) {$\underbrace{\hspace{100pt}}$};
					\node [draw=none] () at (2.5,7.6) {$\underbrace{\hspace{100pt}}$};
					
					\node [draw=none] () at (8.5,7.0) {$(\overline{u}_{\triangleright},\overline{v}_{\triangleright})$};
					\node [draw=none] () at (5.5,6.8) {$(\overline{u}_{\diamond},\overline{v}_{\diamond})$};
					\node [draw=none] () at (2.5,7.0) {$(\overline{u}_{\triangleleft},\overline{v}_{\triangleleft})$};

					\node [draw=none] () at (3.2,10.4) {$\mathit{margin}(\overline{u}_{\triangleleft})$};
					\node [draw=none] () at (5.7,10) {$\mathit{diam}(\overline{u}_{\diamond})$};
					\node [draw=none] () at (8.2,10.4) {$\mathit{margin}(\overline{u}_{\triangleright})$};
					
					\draw[<->] (2,10.1) -- (4,10.1);
					\draw[<->] (7,10.1) -- (9,10.1);
					\draw[<->] (4,9.7) -- (7,9.7);
					
					\node [draw=none] () at (9.8,6) {$\underbrace{\hspace{100pt}}$};
					\node [draw=none] () at (5.5,6) {$\underbrace{\hspace{130pt}}$};
					\node [draw=none] () at (1.2,6) {$\underbrace{\hspace{100pt}}$};
					
					\node [draw=none] () at (9.5,5.4) {Right Strategy};
					\node [draw=none] () at (5.5,5.4) {Middle Strategy};
					\node [draw=none] () at (1.5,5.4) {Left Strategy};
					
					\draw[dashed] (3.1,5.4) -- (3.1,9.2);
					\draw[dashed] (7.9,5.4) -- (7.9,9.2);
				\end{tikzpicture}
			}
\caption{Situation of the Composition Lemma}
\label{fig:comp}
		\end{figure}

\begin{lemma}[Composition Lemma]
Suppose that Duplicator wins the $n$-round games from configurations
$(\overline{u}_{\triangleleft},\overline{v}_{\triangleleft})$ and
$(\overline{u}_{\triangleright},\overline{v}_{\triangleright})$
respectively, and let $D$ be such that Duplicator wins the $3n$-round
$D$-local game from configuration
$(\overline{u}_{\diamond},\overline{v}_{\diamond})$.  If
$\mathit{margin}(\overline{u}_{\triangleleft}) >
2^n$,
$\mathit{margin}(\overline{u}_\triangleright) >
2^n$, $D \geq \mathit{diam}(\overline{u}_\diamond)+2^{n+1}$, and the
corresponding three conditions also hold for $\overline{v}$, then
Duplicator wins the $n$-round game from configuration
$(\overline{u},\overline{v})$.
\label{lem:composition}
 \end{lemma}
\begin{proof}
We show that configuration $(\overline{u},\overline{v})$ is winning
for Duplicator in the $n$-round game.  The proof is by induction on
$n$.  

\emph{Base case $(n=0)$.} Note that $(\overline{u},\overline{v})$ is a
partial isomorphism since
$(\overline{u}_{\triangleleft},\overline{v}_{\triangleleft})$,
$(\overline{u}_{\diamond},\overline{v}_{\diamond})$, and
$(\overline{u}_{\triangleright},\overline{v}_{\triangleright})$ are
all partial isomorphisms,
$\mathit{margin}(\overline{u}_{\triangleleft})$
and
$\mathit{margin}(\overline{v}_{\triangleleft})$
are both greater than one, and likewise for 
$\mathit{margin}(\overline{u}_{\triangleright})$
and
$\mathit{margin}(\overline{v}_\triangleright)$.

Induction step $(n>0)$.  Without loss of generality assume that Spoiler 
plays a move $u'$ in structure $\mathbf{A}$.  We consider three cases.

\emph{Case (i).}  Suppose that $u' <u_l- 2^{n-1}$.  Then Duplicator's
winning strategy in configuration
$(\overline{u}_{\triangleleft},\overline{v}_{\triangleleft})$ yields a
response $v'$ such that
$(\overline{u}_{\triangleleft}u',\overline{v}_{\triangleleft}v')$ is
winning for Duplicator in the $(n-1)$-round game.  In particular,
applying Proposition~\ref{prop:metric}, we have $v' < v_l - 2^{n-1}$.
Applying the induction hypothesis to
$(\overline{u}_{\triangleleft}u',\overline{v}_{\triangleleft}v')$,
$(\overline{u}_{\diamond},\overline{v}_{\diamond})$, and
$(\overline{u}_{\triangleright},\overline{v}_{\triangleright})$ we get
that $(\overline{u}u',\overline{v}v')$ is winning for Duplicator in
the $(n-1)$-round game.

\emph{Case (ii).} Suppose that $u' > u_r+2^{n-1}$.  This case is
entirely analogous to Case (i), except that Duplicator's 
response to $u'$ is generated from her winning strategy in configuration
$(\overline{u}_{\triangleright},\overline{v}_{\triangleright})$.

\emph{Case (iii).} Suppose that $u_l-2^{n-1} \leq u' \leq
u_r+2^{n-1}$.  Then Duplicator's winning strategy in configuration
$(\overline{u}_{\diamond},\overline{v}_{\diamond})$ yields a response
$v'$ such that Duplicator wins the $(3n-1)$-round $D$-local game from
$(\overline{u}_{\diamond}u',\overline{v}_{\diamond}v')$.  By
Proposition~\ref{prop:metric} we must have $v_l-2^{n-1}\leq v' \leq
v_r+2^{n-1}$.

To apply the induction hypothesis, the idea is to ``expand the middle
configuration'' by adding new left and right boundary pebbles
$u'_l,u'_r$ and $v'_l,v'_r$ respectively.  Formally, Spoiler moves
$u'_l:=u_l-2^{n-1}$ and $u'_r:=u_r+2^{n-1}$ in the $D$-local game in
position $(\overline{u}_{\diamond}u',\overline{v}_{\diamond}v')$ force
Duplicator responses $v'_l:=v_l-2^{n-1}$ and $v'_r:=v_r+2^{n-1}$ such
that
$(u'_l\overline{u}_{\diamond}u'u'_r,v'_l\overline{v}_{\diamond}v'v'_r)$
is winning for Duplicator in the $3(n-1)$-round $D$-local game.  By
the same reasoning,
$(\overline{u}_{\triangleleft}u'_l,\overline{v}_{\triangleleft}v'_l)$
and
$(u'_r\overline{u}_{\triangleright},v'_r\overline{v}_{\triangleright})$
are both winning positions for Duplicator in the $(n-1)$-round game.
\emph{A fortiori} $(u_1\ldots u_{l-1}u'_l,v_1\ldots v_{l-1}v'_l)$ and
$(u'_ru_{r+1}\ldots u_s,v'_rv_{r+1}\ldots v_s)$ are also both winning
for Duplicator in the $(n-1)$-round game.  Finally, applying the
induction hypothesis with left configuration $(u_1\ldots
u_{l-1}u'_l,v_1\ldots v_{l-1}v'_l)$, middle configuration
$(u'_l\overline{u}_{\diamond}u'u'_r,v'_l\overline{v}_{\diamond}v'v'_r)$,
and right configuration $(u'_ru_{r+1}\ldots u_s,v'_rv_{r+1}\ldots
v_s)$, we conclude that $(\overline{u}u',\overline{v}v')$ is winning
for Duplicator in $n-1$ rounds.  \qed
\end{proof}

\subsection{3-Variable Theorem}
\begin{proposition}
Suppose that Duplicator wins the $(4n+2)$-round 3-pebble game from a
configuration $(\overline{u},\overline{v})$ with
$|\overline{u}|=|\overline{v}| \leq 3$.  Then she also wins the
$n$-round (unrestricted-pebble) game from configuration
$(\overline{u},\overline{v})$.
\label{prop:3-var}
\end{proposition}
\begin{proof}

	
The proof is by induction on $n$.  The base case $(n=0)$ is immediate,
and the induction step $(n>0)$ is as follows. Suppose that
$|\overline{u}|=|\overline{v}|<3$. Then for any Spoiler move,
Duplicator replies using her 3-pebble strategy, leading to a
3-configuration $(\overline{u}',\overline{v}')$. Duplicator now has a
winning strategy for the $(4n+1)$-round 3-pebble game starting from
the configuration $(\overline{u}',\overline{v}')$, and therefore she
also has a winning strategy for the $(4(n-1)+2)$-round 3-pebble game
from $(\overline{u}',\overline{v}')$. By the induction hypothesis, she
has a winning strategy for the $(n-1)$-round unrestricted game from
$(\overline{u}',\overline{v}')$, and therefore a winning strategy for
the $n$-round game from $(\overline{u},\overline{v})$.

Now suppose that $|\overline{u}|=|\overline{v}|=3$. 
We claim that given any 3-configuration $(\overline{u},\overline{v})$,
we can decompose it into a left part
$(\overline{u}_{\triangleleft},\overline{v}_{\triangleleft})$, a
middle part $(\overline{u}_{\diamond},\overline{v}_{\diamond})$ and a
right part
$(\overline{u}_{\triangleright},\overline{v}_{\triangleright})$,
satisfying the following desiderata:
\begin{enumerate}
	\item $\mathit{diam}(\overline{u}_{\diamond})\leq 2^{n+1}$,
	\item $\mathit{margin}(\overline{u}_{\triangleleft})>2^n$ and $\mathit{margin}(\overline{u}_{\triangleright})>2^n$,
	\item $|\overline{u}_{\triangleleft}| \leq 2$ and $|\overline{u}_{\triangleright}| \leq2$,
	\item Conditions 1--3 hold for $\overline{v}_{\triangleleft}$, $\overline{v}_{\diamond}$, and $\overline{v}_{\triangleright}$.
\end{enumerate}

By Proposition \ref{prop:local}, if the above four conditions hold, we
obtain that Duplicator has a winning strategy for the $3n$-round
$2^{n+2}$-local game from the configuration
$(\overline{u}_{\diamond},\overline{v}_{\diamond})$. Furthermore, by
(3) and the case described above for configurations of size strictly
less than 3, it follows that Duplicator has a winning strategy for the
$n$-round games from the configurations
$(\overline{u}_{\triangleleft},\overline{v}_{\triangleleft})$ and
$(\overline{u}_{\triangleright},\overline{v}_{\triangleright})$.
Thus, by applying the Composition Lemma \ref{lem:composition},
Duplicator has a winning strategy for the $n$-round game from the
configuration $(\overline{u},\overline{v})$.

It remains to show that given any 3-configuration
$(\overline{u},\overline{v})$, we can always find a decomposition that
satisfies the above conditions.  We show this by the following case
analysis.  Without loss of generality, assume that $u_1\leq u_2 \leq u_3$ and
$v_1 \leq v_2 \leq v_3$.

\emph{Case(i).}  Suppose that $u_2-u_1\leq 2^n$ and $u_3-u_2\leq 2^n$. Then it
is also the case that $v_2-v_1\leq 2^n$ and $v_3-v_2\leq 2^n$, since otherwise
Spoiler would have a $n$-round 3-pebble winning strategy by the contraposition of Corollary
\ref{corl:metric}. Then let $\overline{u}_{\triangleleft}=u_1$,
$\overline{u}_{\triangleright}=u_3$ and
$\overline{u}_{\diamond}=u_1u_2u_3$, and assume a corresponding
decomposition of $\overline{v}$.

\emph{Case(ii).} Suppose that $u_3-u_2>2^n$ and $u_2-u_1>2^n$. Then it
is also the case that $v_3-v_2>2^n$ and $v_2-v_1>2^n$ by Corollary
\ref{corl:metric}. Let then $\overline{u}_{\diamond}=u_2$,
$\overline{u}_{\triangleleft}=u_1u_2$,
$\overline{u}_{\triangleright}=u_2u_3$, and consider the corresponding
decomposition for $\overline{v}$.

\emph{Case(iii).} Suppose finally that $u_3-u_2>2^n$ and
$u_2-u_1\leq 2^n$. By Corollary \ref{corl:metric}, we also have that
$v_3-v_2>2^n$ and $v_2-v_1\leq 2^n$. Let
$\overline{u}_{\triangleleft}=u_1$, $\overline{u}_{\diamond}=u_1u_2$,
$\overline{u}_{\triangleright}=u_2u_3$ and consider the corresponding
decomposition of $\overline{v}$.

The case where $u_3-u_2\leq 2^n$ and $u_2-u_1>2^n$ is symmetric.
\qed
\end{proof}

From Proposition~\ref{prop:3-var} and Theorem~\ref{thm:EF} we
immediately obtain our main result:

\begin{theorem}
$(\mathbb{R},<,+1)$ has the 3-variable property.
\label{thm:main}
\end{theorem}

\section{Linear Functions}
\label{sec:arithmetic}
In this section we show the 3-variable property for
the $\sigma$-structure $(\mathbb{R},<,f)$ with
$f:\mathbb{R}\rightarrow\mathbb{R}$ a linear function $f(x)=ax+b$.
This follows fairly straightforwardly from our main result,
Theorem~\ref{thm:main}, using the classical compositional method for
sums of ordered structures.

\subsection{Monotone Linear Functions}
\label{sec:monotone}
Consider $f:\mathbb{R}\rightarrow \mathbb{R}$ given by $f(x)=ax+b$,
where $a,b\in \mathbb{R}$ and $a>0$.  We prove that $(\mathbb{R},<,f)$
has the 3-variable property.

Suppose that $a=1$, that is, $f(x)=x+b$.  If $b>0$ then
$(\mathbb{R},<,f)$ is isomorphic to $(\mathbb{R},<,+1)$.  If $b<0$
then $(\mathbb{R},<,f)$ is isomorphic to
$(\mathbb{R},<^{\mathrm{op}},+1)$, where $<^{\mathrm{op}}$ is the
opposite order on $\mathbb{R}$.  In either case $(\mathbb{R},<,f)$
inherits the 3-variable property from $(\mathbb{R},<,+1)$.

Assume now that $a\neq 1$. Notice that $f$ has a unique fixed point
$x^*=\frac{b}{1-a}$.  Moreover, considering the intervals
$I_0=(-\infty,x^*)$ and $I_1=(x^*,\infty)$, $f$ restricts to bijections
$f_i:I_i\rightarrow I_i$ for $i=0,1$.  Now the map
$\Phi_0(x)=-\log(x^*-x)$ defines an isomorphism of $\sigma$-structures
from $(I_0,<,f_0)$ to $(\mathbb{R},<,+a)$.  Likewise the map
$\Phi_1(x)=\log(x-x^*)$ defines an isomorphism from $(I_1,<,f_1)$ to
$(\mathbb{R},<,+a)$.  It follows that $(I_0,<,f_0)$ and $(I_1,<,f_1)$
both have the 3-variable property.

We argue that $(\mathbb{R},<,f)$ has the 3-variable property as
follows.  Let $\mathbf{A}$ and $\mathbf{B}$ be expansions of
$(\mathbb{R},<,f)$ with interpretations of the monadic predicate
variables.  Let $\mathbf{A}_0$ be the sub-structure of $\mathbf{A}$
with domain $I_0$ and let $\mathbf{A}_1$ be the sub-structure of
$\mathbf{A}$ with domain $I_1$.  Define $\mathbf{B}_0$ and
$\mathbf{B}_1$ likewise.  Then if Spoiler wins the $n$-round EF game
on $\mathbf{A}$ and $\mathbf{B}$ he also wins the $n$-round
game on the substructures $\mathbf{A}_0$ and $\mathbf{B}_0$ and the
$n$-round game on $\mathbf{A}_1$ and $\mathbf{B}_1$.  Thus there
exists $m$, depending only on $n$, such that Spoiler wins the
$m$-round 3-pebble EF games on $\mathbf{A}_0$ and $\mathbf{B}_0$ and
on $\mathbf{A}_1$ and $\mathbf{B}_1$.  Then by the usual composition
argument on sums of ordered structures~\cite{ImmermanK89}, we can show
that Spoiler wins the $m$-round 3-pebble game on $\mathbf{A}$ and
$\mathbf{B}$.

\subsection{Antitone Linear Functions}
\label{sec:antitone}
Consider a linear function $f(x)=ax+b$, where $a<0$.  Note that the
map $f^2:=f\circ f : \mathbb{R}\rightarrow \mathbb{R}$ is monotone and
linear.  The idea is to exploit the fact that $(\mathbb{R},<,f^2)$ has
the 3-variable property to rewrite a given monadic first-order
$\sigma$-sentence $\varphi$ to a 3-variable sentence $\varphi''$ that
is equivalent to $\varphi$ over $(\mathbb{R},<,f)$.  In this rewriting
it is convenient to use $x^*$ as an additional constant symbol in
intermediate forms, where $x^*$ is the unique fixed point of $f$. We
also allow nested applications of $f$ in intermediate formulas.

We obtain $\varphi''$ as follows.  Motivated by the fact that $f$ maps
the open interval $(x^*,\infty)$ onto $(-\infty,x^*)$ and \emph{vice
  versa}, working bottom-up, replace each subformula $\exists x\,\psi$
by
\[\exists x\, (x>x^* \wedge (\psi \vee \psi[f(x)/x] \vee \psi[x^*/x])) \, .\] 

Now simplify the atomic subformulas as follows, bearing in mind that
all variables range over $(x^*,\infty)$.  Replace every term
$f^n(x^*)$ with $x^*$.  Replace $f^n(x)=x^*$ with \textbf{false}.
Replace $f^n(x)=f^m(y)$ with $f^{n-1}(x)=f^{m-1}(y)$ if $n$ and $m$
are both odd, and with \textbf{false} if $n$ and $m$ have different
parity.  If $n$ is odd then replace $x^*<f^n(x)$ with \textbf{false}
and $f^n(x)<x^*$ with \textbf{true}.  Replace $f^n(x)<f^m(y)$ by
$f^{n-1}(x)<f^{m-1}(y)$ if $n$ and $m$ are both odd, by \textbf{true}
if $n$ is odd and $m$ is even, and by \textbf{false} if $n$ is even
and $m$ is odd.  Finally eliminate the constant symbol $x^*$ using the
fact that it is definable in terms of $f^2$, e.g., replace each
subformula $P(x^*)$ with $\exists y\,(y=f^2(y) \wedge P(y))$.

Let $\varphi'$ denote the sentence arising from the above
transformation.  Treating the atomic formulas $P(f(x))$ as unary
predicate variables, we can interpret $\varphi'$ as a monadic
first-order sentence over the structure $(\mathbb{R},<,f^2)$.  Since
$f^2$ is monotone we can use the result of Section~\ref{sec:monotone}
to transform $\varphi'$ to an equivalent 3-variable sentence
$\varphi''$ over $(\mathbb{R},<,f^2)$.  Then $\varphi''$ is equivalent
to $\varphi$ considered as a formula over the structure
$(\mathbb{R},<,f)$.

%
%

\section{Counterexample}
\label{sec:counterexample}
In this section we exhibit a countable dense linear order
$E$ and function $g:E \rightarrow E$ such that $(E,<,g)$ does not have
the $k$-variable property for any $k$.

Let $(S,<)$ be the set of non-empty finite sequences of integers under
the lexicographic order, and let $E$ be the equivalence relation on
$S$ that relates any two such sequences that end with the same
element.  Since the integers have no greatest or least element, any
non-empty interval in $S$ contains an element of each $E$-equivalence
class.  Venema~\cite{Venema90} has shown that the structure $(S,<,E)$
does not have the $k$-variable property for any $k$.  For example, one
can express the property \emph{``predicate $P$ holds on at least $k+1$
  $E$-inequivalent elements''} with $k+1$ variables but not $k$
variables. Indeed it is not hard to see that in the $k$-pebble EF game
(over any number of rounds) Spoiler cannot distinguish the cases that
predicate $P$ is a union of $k$ $E$-equivalence classes and that $P$
is a union of $k+1$ $E$-equivalence classes.

We next translate this example to the setting of linear orders with
unary functions.  Consider the equivalence relation $E$ above as an
ordered set under the lexicographic order on $S\times S$.  Define
$g:E\rightarrow E$ by $g(s,t)=(t,s)$ and consider the
$\sigma$-structure $\mathbf{E}=(E,<,g)$ (where $\sigma$ is the
signature for linear orders and unary functions, defined in
Section~\ref{sec:EF}).  Note that $g$ is very far from being monotone.

To each labelled expansion $\mathbf{S}$ of $(S,<,E)$ we associate a
labelled expansion $\mathbf{E}$ of $(E,<,g)$, where $P^{\mathbf{E}} =
\{(s,s):s\in P^{\mathbf{S}}\}$ for each monadic predicate symbol $P$.
There is moreover a one-dimensional interpretation $\Gamma$
(cf.\ Section~\ref{sec:interpretation}) of $\mathbf{S}$ in
$\mathbf{E}$.  The domain formula $\partial_\Gamma(x)$ of $\Gamma$ is
$x=g(x)$ so that $\partial_\Gamma (E)= \{(s,t) \in E:s=t\}$.  The
coding map $f_\Gamma:\partial_\Gamma (E) \rightarrow S$ is given by
$f_\Gamma(s,s)=s$.  The interpretation also specifies for each atomic
formula $\varphi(x_1,\ldots,x_m)$ over $\mathbf{S}$ a corresponding
formula $\varphi_\Gamma(x_1,\ldots,x_m)$ over
$\mathbf{E}$, with $\mathbf{S} \models \varphi[s_1,\ldots,s_m]$ if and
only if $\mathbf{E} \models
\varphi_\Gamma[(s_1,s_1),\ldots,(s_m,s_m)]$ for all $s_1,\ldots,s_m
\in S$.  This correspondence sends $x<y$ and $P(x)$ to themselves and
$E(x,y)$ to the formula $\psi(x,y)\vee\psi(y,x)$, where
\begin{align*}\psi(x,y)\,:=\, & \exists u\,(x<u<g(u)<y \, \wedge \\
&\forall v (x<v<u  \vee  g(u)<v<y
\rightarrow g(v)\neq v)) \, .
\end{align*}


Conversely there is a natural two-dimensional first-order
interpretation $\Gamma$ of $\mathbf{E}$ in $\mathbf{S}$.  The domain
formula is $\partial_\Gamma(x,y) = E(x,y)$, and thus
$\partial_\Gamma(S^2) = \{(s,t) \in S\times S : (s,t)\in E\}$.  The coding map
$f_\Gamma:\partial_\Gamma(S^2)\rightarrow E$ is given by
$f_\Gamma(s,t)=(s,t)$. The translation of atomic formulas over
$\mathbf{E}$ to corresponding formulas over $\mathbf{S}$ is similarly
straightforward, e.g., $x<y$ is mapped to $x_1<y_1 \vee(x_1=y_1
\wedge x_2<y_2)$.

As observed in Dawar~\cite[Section 3]{Dawar05} in a similar context,
the existence of such a two-way interpretation entails that if
$(E,<,g)$ has the $k$-variable property for some $k$ then $(S,<,E)$
has the $k'$-variable property for some $k'$.  It follows that
$(E,<,g)$ does not have the $k$-variable property for any $k$.

\section{Conclusion and Future Work}
We have shown that the structure $(\mathbb{R},<,f)$ has the 3-variable
property for linear functions $f:\mathbb{R}\rightarrow \mathbb{R}$.
In future work it would be natural to consider whether the
$k$-variable property holds, for some $k$, for richer classes of functions,
e.g., classes of polynomials.

Moving beyond the reals, we would like to explore whether the results
in this paper generalise to arbitrary linear orders and families of
monotone functions thereon.  More generally, there is the problem,
raised by Immerman and Kozen in the conclusion of~\cite{ImmermanK89},
of finding a model-theoretic characterisation of those classes of
structures possessing the $k$-variable property for some $k$.

In those settings in which the $k$-variable property holds, following
\cite{GroheS05}, it is natural to consider how the number of variables
affects the succinctness of formulas and, in view of~\cite{Gabbay81},
also to seek expressively complete temporal logics.

\bibliographystyle{plain}
\bibliography{metric}

\begin{thebibliography}{10}

\bibitem{Dawar05}
Anuj Dawar.
\newblock How many first-order variables are needed on finite ordered
  structures?
\newblock In {\em We Will Show Them! Essays in Honour of Dov Gabbay, Volume
  One}, pages 489--520. College Publications, 2005.

\bibitem{Gabbay81}
D.~M. Gabbay.
\newblock Expressive functional completeness in tense logic.
\newblock In U.~M\"onnich, editor, {\em Aspects of Philosophical Logic}, pages
  91--117. Reidel, Dordrecht, 1981.

\bibitem{GabbayPSS80}
Dov~M. Gabbay, Amir Pnueli, Saharon Shelah, and Jonathan Stavi.
\newblock On the temporal basis of fairness.
\newblock In {\em POPL}, pages 163--173. ACM Press, 1980.

\bibitem{GroheS05}
Martin Grohe and Nicole Schweikardt.
\newblock The succinctness of first-order logic on linear orders.
\newblock {\em Logical Methods in Computer Science}, 1(1), 2005.

\bibitem{HirshfeldR12}
Yoram Hirshfeld and Alexander Rabinovich.
\newblock Continuous time temporal logic with counting.
\newblock {\em Inf. Comput.}, 214:1--9, 2012.

\bibitem{HirshfeldR05}
Yoram Hirshfeld and Alexander~Moshe Rabinovich.
\newblock Timer formulas and decidable metric temporal logic.
\newblock {\em Inf. Comput.}, 198(2):148--178, 2005.

\bibitem{Hodges}
Wilfrid Hodges.
\newblock {\em A Shorter Model Theory}.
\newblock Cambridge University Press, New York, NY, USA, 1997.

\bibitem{HodkinsonS97}
Ian Hodkinson and Andr\'a{}s Simon.
\newblock The {$k$}-variable property is stronger than {$H$}-dimension {$k$}.
\newblock {\em Journal of Philosophical Logic}, 26(1):81--101, 1997.

\bibitem{Hunter13}
Paul Hunter.
\newblock When is metric temporal logic expressively complete?
\newblock In {\em CSL}, volume~23 of {\em LIPIcs}, pages 380--394. Schloss
  Dagstuhl, 2013.

\bibitem{HunterOW13}
Paul Hunter, Jo{\"e}l Ouaknine, and James Worrell.
\newblock Expressive completeness for metric temporal logic.
\newblock In {\em LICS}, pages 349--357. IEEE Computer Society, 2013.

\bibitem{Immerman82}
Neil Immerman.
\newblock Upper and lower bounds for first order expressibility.
\newblock {\em J. Comput. Syst. Sci.}, 25(1):76--98, 1982.

\bibitem{ImmermanK89}
Neil Immerman and Dexter Kozen.
\newblock Definability with bounded number of bound variables.
\newblock {\em Inf. Comput.}, 83(2):121--139, 1989.

\bibitem{Kamp68}
H.~Kamp.
\newblock {\em Tense Logic and the Theory of Linear Order}.
\newblock PhD thesis, University of California, 1968.

\bibitem{Poizat82}
Bruno Poizat.
\newblock Deux ou trois choses que je sais de {$L_n$}.
\newblock {\em J. Symb. Log.}, 47(3):641--658, 1982.

\bibitem{Rossman08}
Benjamin Rossman.
\newblock On the constant-depth complexity of k-clique.
\newblock In {\em STOC}, pages 721--730. ACM, 2008.

\bibitem{Venema90}
Yde Venema.
\newblock Expressiveness and completeness of an interval tense logic.
\newblock {\em Notre Dame Journal of Formal Logic}, 31(4):529--547, 1990.

\end{thebibliography}

\newpage

\begin{appendix}
	\section{Appendix}
	
	\subsection{Missing Proofs from Section \ref{sec:local}}

	\begin{propositionAppendix}{prop:relative-frac}
	Let $u_1\ldots u_s \equiv v_1\ldots v_s$ be
	two assignments. Then 
	\[ \mathrm{frac}(u_i-u_k)<\mathrm{frac}(u_j-u_k)
	   \;\Leftrightarrow\;
	   \mathrm{frac}(v_i-v_k)<\mathrm{frac}(v_j-v_k)
	\]
	for all indices $i,j,k \in \{1,\ldots,s\}$.
	\end{propositionAppendix}

	\begin{proof}
Fix $i,j,k \in \{1,\ldots,s\}$.
From the assumption $u_1\ldots u_s \equiv v_1\ldots v_s$ we 
have the following chain of equivalences:
\begin{eqnarray*}
\mathrm{frac}(u_i-u_k)<\mathrm{frac}(u_j-u_k) &\,\Leftrightarrow\,&
u_i-u_k-\lfloor u_i-u_k\rfloor < u_j-u_k-\lfloor u_j-u_k\rfloor\\
&\Leftrightarrow&
 u_i-u_j < \lfloor u_i-u_k\rfloor-\lfloor u_j-u_k\rfloor\\
&\Leftrightarrow&
 v_i-v_j < \lfloor u_i-u_k\rfloor-\lfloor u_j-u_k\rfloor\\
&\Leftrightarrow&
 v_i-v_j < \lfloor v_i-v_k\rfloor-\lfloor v_j-v_k\rfloor\\
&\Leftrightarrow&
v_i-v_k-\lfloor v_i-v_k\rfloor < v_j-v_k-\lfloor v_j-v_k\rfloor\\
&\Leftrightarrow&
\mathrm{frac}(v_i-v_k)<\mathrm{frac}(v_j-v_k) \, .
 \end{eqnarray*}
\qed
\end{proof}
%

\noindent \begin{propositionAppendix}{prop:composition}
	Suppose that $u_1 \ldots u_s$ and $v_1 \ldots v_s$ are both increasing
	and that $u_1 \ldots u_m \equiv v_1 \ldots v_m$ and
	$u_m \ldots u_s \equiv v_m \ldots v_s$ for some $m$, $1 \leq
	m \leq s$.  Then $u_1 \ldots u_s \equiv v_1 \ldots v_s$.

	\end{propositionAppendix} 

\begin{proof} We must show that
	$\lfloor u_j-u_i \rfloor = \lfloor v_j-v_i \rfloor$ for all
	$i \leq m < j$.  To this end, we observe that since $u_1\ldots
	u_s$ is increasing,
	\begin{eqnarray*}
	\mathrm{frac}(u_j-u_i) &=& \mathrm{frac}(u_j-u_m+(u_m-u_i))\\
	&=& \mathrm{frac}(u_j-u_m) + \mathrm{frac}(u_m-u_i) \, .
	\end{eqnarray*}
	It follows that
	\begin{eqnarray*}
	\lfloor u_j-u_i \rfloor &=& u_j-u_i - \mathrm{frac}(u_j-u_i)\\
	  &=& (u_j-u_m) + (u_m-u_i) - (\mathrm{frac}(u_j-u_m) + \mathrm{frac}(u_m-u_i))\\
&=& (u_j-u_m) - \mathrm{frac}(u_j-u_m) + (u_m-u_i) - \mathrm{frac}(u_m-u_i)\\
	  &=& \lfloor u_j-u_m \rfloor + \lfloor u_m-u_i \rfloor \, .
          \end{eqnarray*}
We can similarly show that 
\[ \lfloor v_j-v_i \rfloor = \lfloor v_j-v_m \rfloor + \lfloor v_m-v_i \rfloor \, .\]
But $\lfloor u_j-u_m \rfloor = \lfloor v_j-v_m \rfloor$ since
$u_m \ldots u_s \equiv v_m \ldots v_s$.  Likewise $\lfloor
u_m-u_i \rfloor = \lfloor v_m-v_i \rfloor$ since $u_1 \ldots
u_m \equiv v_1 \ldots v_m$.  We conclude that $\lfloor u_j-u_i \rfloor
= \lfloor v_j-v_i \rfloor$.
\qed
        \end{proof}

\end{appendix}

\end{document}